[2023/10/19]
\documentclass[draft]{article}
\title{A finite dimensional trace formula}
\author{Zhao Tianhong}
\date{2025/5/18}

\usepackage{amsfonts}
\usepackage{amsmath}
\usepackage{amssymb}
\usepackage{mathtools}
\usepackage{amsthm}
\usepackage{pgfplots}
\usepackage{enumitem}
\usepackage{tikz} 
\usepackage{tikz-cd}
\usepackage{changepage}

\chardef\bslash=`\\ 

\newtheorem{thm}{Theorem}[section]
\newtheorem{cor}[thm]{Corollary}
\newtheorem{lem}[thm]{Lemma}

\newtheorem{ex}{Example}[section]

\theoremstyle{definition}

\theoremstyle{remark}

\newcommand{\R}{\mathbb{R}}
\newcommand{\Z}{\mathbb{Z}}

\begin{document}
	\maketitle
\begin{abstract}
We take the trace of Von-Neumann's ergodic theorem and get a trace formula of a unitary matrix family. It is an extension of Poisson summation formula in higher dimension. We also construct a family of crystalline measure with complex coefficient.
\end{abstract}

\section{Introduction}
  It is known that crystalline measure\cite{cry}\cite{meyer2023crystalline} can be constructed by Lee-Yang Polynomial.

 We give a Von-Neumann's ergodic theorem version of trace formula\cite{lenz2009continuity}\cite{silva1998trace} of finite dimensional unitary matrix and give a example of crystalline measure with complex coefficient.

Let $U\left(k\right),k\in \R$ be a family of finite-dimensional unitary matrix, $A\left(k\right),k\in \R$ be a family of uniform bounded matrix. Suppose $k \mapsto U\left(k\right)$ and $k \mapsto A\left(k\right)$ are smooth, with some growth condition. Suppose $f\left(k\right)=\det\left(U\left(k\right)-I\right)$ only has simple zero. Let $\lambda\left(k\right)$ be the eigenvalue of $U\left(k\right)$, $\left| \lambda\left(k\right) \right\rangle$ be the eigenvector of $\lambda\left(k\right)$, then we have:
\begin{thm}
	$$\frac{1}{2\pi}\sum_{m\in \Z}tr\left(U^m\left(k\right)A\left(k \right)  \right)= \sum_{k_0} \frac{\left \langle \lambda\left (k_0\right )  |A\left (k_0\right ) |\lambda\left (k_0\right )   \right \rangle }{\left | \left \langle \lambda\left (k_0\right )  |U^{'}\left (k_0\right ) |\lambda\left (k_0\right )   \right \rangle \right | }\delta\left(k-k_0\right) $$
	where $k_0\in \R$ satisfy: $\det\left(U\left(k_0\right)-I\right)=0, \ \ \lambda\left(k_0 \right)=1. $
\end{thm}
Let $A\left(k\right)=U^{'}\left(k\right)=\frac{dU\left(k\right)}{dk}$, suppose $U\left(k\right) $ has positive generator, then we have:
\begin{thm}
	$$\frac{1}{2\pi i}\sum_{m\in \Z}tr\left(U^m\left(k\right)U^{'}\left(k \right)  \right)= \sum_{k_0} \delta\left(k-k_0\right) $$
\end{thm}

\section{Proof}
We start at one-dimensional case:
\begin{thm}(Poisson summation formula)
	Let $z \in S^1$, we have:
	$$
	\frac{1}{2\pi}\sum_{m\in \Z}z^m=\delta\left(z-1\right).
	$$
\end{thm}
To get the usual Poisson summation formula, we only need to take $z=e^{2\pi it}, t\in \R$ and use the change the variable of Dirac $\delta$ function:
$$
\delta\left(e^{2\pi it}-1\right)=\frac{1}{2\pi}\sum_{m} \delta\left(t-m\right),
$$
then we will get:
$$
\sum_{m\in \Z}e^{2m\pi it}=\sum_{m\in \Z} \delta\left(t-m\right).
$$
Now we extend this statement to high-dimensional case.

We want to consider a summation: 
$$\frac{1}{2\pi}\sum_{m\in \Z}U^{m},$$
where $U$ is a unitary matrix.

From Von-Neumann's ergodic theorem, we know:
$$
\lim_{N\to \infty}\frac{1}{2N}\sum_{m=-N}^{N}U^{m}=P_1,
$$
$P_1$ is the projection on invariant subspace:
$$
X_{inv}=\left\lbrace x\in X, Ux=x \right\rbrace,
$$
and $X$ is the underlying Hilbert space. So it could be better to consider the trace of $\frac{1}{2\pi}\sum_{m\in \Z}U^{m}$ and get the information of fixed subspace by $U$. Noted that $P_1=0$ if and only if $U$ does not have eigenvalue $1$.
\begin{thm}
	Let $U\left(k\right),\ k\in\R$ be a family of unitary matrix, $k \mapsto U\left(k\right)$ is smooth. Suppose $f\left(k\right)=\det\left(U\left(k\right)-I\right)$ only has simple zero and the eigenvalue $\lambda\left(k\right)$ has non-zero speed: $\frac{d\lambda}{dk}\neq 0$ and the fourier coefficient $c_m\left(g\right) =\int_{-\infty}^{\infty}tr\left( U^m\left(k\right)\right)g\left(k\right)dk$ satisfy: $\sum_{m\in \Z}c_m\left(g\right)$ converges absolutely, where $g$ is Schwartz function.
	
	Then we have:
	$$
	\frac{1}{2\pi}\sum_{m\in \Z}tr\left( U^m\left(k\right)\right)  =\sum_{\lambda}\delta\left(\lambda\left(k\right) -1\right)=\sum_{k_0}\frac{\delta\left(k-k_0\right).}{\left|\frac{d\lambda}{dk}(k_0)\right| }
	$$
	where $\lambda\left(k\right)$ is the eigenvalue of $U\left(k\right)$, the middle sum runs over all the eigenvalue of $U\left(k\right)$, the right sum runs over all the $k_0\in \R$ such that:
	$$f\left(k_0\right)=\det\left(U\left(k_0\right)-I\right)=0.$$
	The sum converges in the sense of distribution on Schwartz space $\mathcal{S}\left(\R\right) $.
\end{thm}
\begin{proof}
	Let $-1<t<1$ be a parameter, $g\left(k\right)\in \mathcal{S}\left(\R\right)$, then we have:
	\begin{align*}
	 \frac{1}{2\pi}\sum_{m\in \Z}t^{\left|m\right| }\int_{-\infty}^{\infty}tr\left(U^m\left(k\right)\right)g\left (k\right )dk=& 	\frac{1}{2\pi}\int_{-\infty}^{\infty}\sum_{m\in \Z}t^{\left|m\right| }tr\left(U^m\left(k\right)\right)g\left (k\right )dk \\
	 =&\frac{1}{2\pi}\sum_{\lambda}\int_{-\infty}^{\infty}\frac{1-t^2}{\left | t-\lambda\left (k\right )  \right |^2 } g\left (k\right )dk
	\end{align*}
	First "=" is dominated convergence theorem. From definition of Dirac $\delta$ function:
	$$
	t\to 1^{-},\ \ \ \frac{1}{2\pi}\frac{1-t^2}{\left | t-\lambda\left (k\right )  \right |^2} \to \delta\left(\lambda\left(k\right) -1\right).
	$$
	Take the limit: $t\to 1^{-}$ and we get:
	$$
	\frac{1}{2\pi}\sum_{m\in \Z}tr\left( U^m\left(k\right)\right)  =\sum_{\lambda}\delta\left(\lambda\left(k\right) -1\right).
	$$
	Suppose $f\left(k\right)=\det\left(U\left(k\right)-I\right)$ only has simple zero and the eigenvalue $\lambda\left(k\right)$ has non-zero speed: $\frac{d\lambda}{dk}\neq 0$, then from change of variable we get:
	$$
	\sum_{\lambda}\delta\left(\lambda\left(k\right) -1\right)=\sum_{k_0}\frac{\delta\left(k-k_0\right).}{\left|\frac{d\lambda}{dk}(k_0)\right| }
	$$
	where $\lambda\left(k\right)$ is the eigenvalue of $U\left(k\right)$, the left sum runs over all the eigenvalue of $U\left(k\right)$, the right sum runs over all the $k_0\in \R$ such that:
    $$f\left(k_0\right)=\det\left(U\left(k_0\right)-I\right)=0.$$
\end{proof}
 
\begin{lem}(speed of eigenvalue)
	Suppose $f\left(k\right)=\det\left(U\left(k\right)-I\right)$ only has simple zero.
	We have: $$\frac{d\lambda}{dk}\left(k_0\right)=\left \langle \lambda\left (k_0\right )  |U^{'}\left (k_0\right ) |\lambda\left (k_0\right )   \right \rangle.$$ for all the $k_0\in \R$ such that $$f\left(k_0\right)=\det\left(U\left(k_0\right)-I\right)=0.$$ 
\end{lem}
\begin{proof}
Take the differential of: 
$$
\lambda(k)=\left\langle\lambda(k)| U(k)|\lambda(k)\right\rangle=tr\left(P_{\lambda}U\right)   
$$
and use the smoothness of $U\left( k\right) $, $\lambda\left(k_0\right) =1$, we can get the result.
\end{proof}

\begin{thm}
	Suppose $A\left(k\right), k\in \R$ is a family of uniform bounded matrix,	Let $U\left(k\right),\ k\in\R$ be a family of unitary matrix. Suppose $k \mapsto U\left(k\right)$ and $k \mapsto A\left(k\right)$ are smooth. Suppose $f\left(k\right)=\det\left(U\left(k\right)-I\right)$ only has simple zero and the eigenvalue $\lambda\left(k\right)$ has non-zero speed: $\frac{d\lambda}{dk}\neq 0$ and the fourier coefficient $c_m\left(g\right) =\int_{-\infty}^{\infty}tr\left( U^m\left(k\right)A\left(k\right) \right)g\left(k\right)dk$ satisfy: $\sum_{m\in \Z}c_m\left(g\right)$ converges absolutely, where $g$ is Schwartz function.
		
	Then we have:
	$$\frac{1}{2\pi}\sum_{m\in \Z}tr\left(U^m\left(k\right)A\left(k \right)  \right)= \sum_{k_0} \frac{\left \langle \lambda\left (k_0\right )  |A\left (k_0\right ) |\lambda\left (k_0\right )   \right \rangle }{\left | \left \langle \lambda\left (k_0\right )  |U^{'}\left (k_0\right ) |\lambda\left (k_0\right )   \right \rangle \right | }\delta\left(k-k_0\right) $$
where $k_0\in \R$ satisfy: $\det\left(U\left(k_0\right)-I\right)=0, \ \ \lambda\left(k_0 \right)=1. $
	The sum converges in the sense of distribution on Schwartz space $\mathcal{S}\left(\R\right) $.
\end{thm}
\begin{proof}
	Similar as Theorem 2.2, we have:
\end{proof}
\begin{align*}
	\sum_{m\in \Z}t^{\left|m\right| }\int_{-\infty}^{\infty}tr\left(U^m\left(k\right)A\left(k \right) \right)g\left (k\right )dk=& 	\int_{-\infty}^{\infty}\sum_{m\in \Z}t^{\left|m\right| }tr\left(U^m\left(k\right)A\left(k \right) \right)g\left (k\right )dk \\
	=\sum_{\lambda}&\int_{-\infty}^{\infty}\frac{(1-t^2)\left \langle \lambda\left (k\right )  |A\left (k\right ) |\lambda\left (k\right )   \right \rangle}{\left | t-\lambda\left (k\right )  \right |^2 } g\left (k\right )dk.
\end{align*}
Take the limit: $t\to 1^{-}$, we get the result.
\begin{cor}
	Suppose the generator of $U\left(k\right)$: 
	$$D\left(k\right)=-i\frac{dU(k)}{dk}U^{-1}\left(k\right)$$
	is uniform-bounded positive definite matrix satisfy the condition above, then we have:
\begin{thm}
	$$\frac{1}{2\pi i}\sum_{m\in \Z}tr\left(U^m\left(k\right)U^{'}\left(k \right)  \right)= \sum_{k_0} \delta\left(k-k_0\right) $$
\end{thm}
\end{cor}
\begin{ex}(Crystalline measure)
	
	Suppose $U\left(k\right)=\begin{pmatrix}
		e^{ik\ell_1}  & 0 & .. & 0\\
		0 & e^{ik\ell_2} & .. & 0\\
		0 & 0 & ..& 0\\
		.. & .. & ..& e^{ik\ell_n}
	\end{pmatrix}S,\ $
	where S is a constant unitary matrix, $\ell_1,\ell_2,...,\ell_n>0$. Let $A$ be a constant unitary matrix. Then the trace formula can be written as:
	$$
	\sum_{n\in \Z}\sum_{n_1+n_2...+n_N=n} c\left(n_1,n_2,...n_N\right) e^{ik\left(n_1\ell_1+...+n_N\ell_N\right) }=\sum_{k_0}c'\left(n\right) \delta(k-k_0)
	$$
	where $k_0$ is zero of $f\left(k\right)=det\left(I-U\left(k\right)  \right)$, $c'\left(n\right)=\frac{\left \langle \lambda\left (k_0\right )  |A\left (k_0\right ) |\lambda\left (k_0\right )   \right \rangle }{\left | \left \langle \lambda\left (k_0\right )  |U^{'}\left (k_0\right ) |\lambda\left (k_0\right )   \right \rangle \right | }$, $n_1,n_2...n_N$ run over all non-positive numbers or all non-negetive numbers.
\end{ex}

\section{Visible of this result}
 There are two ways to draw this result by software. First, we may consider the sum: $\sum_{m\in \Z}t^{\left|m\right| }tr\left(U^m\left(k\right)\right), -1<t<1$, take $t$ near 1.
 
 Second, we can use Newton's identity. Let $h\left(\lambda,k\right)=\det\left(\lambda I-U\left(k\right) \right) $ be the characteristic polynomial of $U\left(k\right)$, from Cayley-Hamilton theorem we have:
 $$
 h\left(U,k\right)=\sum_{i=0}^{n} c_i\left ( k \right ) U^{i}\left (k\right ) =0,
 $$
 multiply $U$ and $U^{-1}$ and take the trace we will get many linear relationships, known as Newton's identity\cite{silva1998trace}:
 
 $$
 \begin{cases}
 	...  \\
 	\sum_{i=0}^{n} c_i\left ( k \right )tr\left  ( U^{i+j}\left (k\right )\right )&=0  \\
 	\sum_{i=0}^{n} c_i\left ( k \right )tr\left  ( U^{i+j+1}\left (k\right )\right )&=0 \\
 	\sum_{i=0}^{n} c_i\left ( k \right )tr\left  ( U^{i+j+2}\left (k\right )\right )&=0 \\
 	...
 \end{cases} 
 $$
 where $j \in \Z$ is an arbitary integer.
 
 Using the Newton's identity, we have:
 $$
 \left( \sum_{i=0}^{n}c_i\left(k\right)\right)\left(\sum_{j=-N}^{N}tr\left  ( U^{j}\left (k\right )\right )\right)=R\left(n,N,k\right),    
 $$ 
 where the reminder $R\left(n,N,k\right)$ is uniform bounded on variable $N,k$. Thus we have:
 $$
 \left(\sum_{j=-N}^{N}tr\left  ( U^{j}\left (k\right )\right )\right)=\frac{R\left(n,N,k\right)}{h\left(1,k\right)}=\frac{R\left(n,N,k\right)}{\det\left(I-U\left(k\right) \right) }
 $$

\bibliographystyle{plain}
\bibliography{Ref}

\end{document}